\documentclass{eptcs}

\usepackage{breakurl}
\usepackage{amsthm}
\usepackage{amssymb}
\usepackage{amsmath}
\usepackage{enumitem}
\usepackage{hyperref}
\usepackage{stmaryrd}
\usepackage{semantic}
\usepackage{xcolor}

\usepackage{fancyvrb}
\CustomVerbatimEnvironment{LLVerbatim}{Verbatim}{fontsize=\small,xleftmargin=0.5cm,xrightmargin=0.5cm,frame=lines,framesep=3mm,commandchars=\\\{\}}
\RecustomVerbatimEnvironment{Verbatim}{Verbatim}{fontsize=\footnotesize,xleftmargin=0.5cm,xrightmargin=0.5cm,frame=lines,framesep=3mm,commandchars=\\\{\}}
\CustomVerbatimEnvironment{SVerbatim}{Verbatim}{fontsize=\footnotesize,xleftmargin=0.3cm,xrightmargin=0.5cm,framesep=3mm,commandchars=\\\{\}}
\CustomVerbatimEnvironment{CVerbatim}{Verbatim}{fontsize=\small,xleftmargin=0.5cm,xrightmargin=0.5cm,framesep=3mm,frame=lines}

\hypersetup{%
  colorlinks=false,
  linkbordercolor=red,
  pdfborderstyle={/S/U/W 1}
}
\allowdisplaybreaks

\newcounter{block}
\newtheorem{theorem}{Theorem}
\newtheorem{lemma}{Lemma}
\theoremstyle{definition}

\newtheorem{definition}[block]{Definition}
\newcommand{\rulec}[1]{\text{\footnotesize{(#1)}}}

\newcommand{\ie}{\emph{i.e.}}
\newcommand{\eg}{\emph{e.g.}}

\newcommand{\interp}[1]{\llbracket{#1}\rrbracket}
\newcommand{\interpA}[1]{\llparenthesis{\,#1\,}\rrparenthesis}

\newcommand{\synLet}[3]{\mathbf{let} \, #2 \,\leftarrow #1\, \mathbf{in} \, #3}
\newcommand{\letb}{\textbf{let}}

\newcommand{\unit}{\textbf{unit}}
\newcommand{\nat}{\mathbf{nat}}

\newcommand{\getS}{\mathsf{get}}
\newcommand{\putS}{\mathsf{put}}
\newcommand{\stopS}{\mathsf{stop}}
\newcommand{\getF}{\mathbf{G}}
\newcommand{\putF}{\mathbf{P}}
\newcommand{\getO}{\mathsf{get}}
\newcommand{\putO}{\mathsf{put}}
\newcommand{\sucO}{\mathsf{suc}}
\newcommand{\unitO}{\mathsf{unit}}
\newcommand{\effChan}{\mathit{eff}}

\newcommand{\ei}{\mathsf{ei}}
\newcommand{\eo}{\mathsf{eo}}
\newcommand{\ea}{\mathsf{ea}}
\newcommand{\eb}{\mathsf{eb}}

\newcommand{\du}[1]{\overline{#1}}

\newcommand{\eff}[1]{\,#1} 
\newcommand{\trule}[1]{\text{{\small{({#1})}}}\;}

\newif\ifwebpage
\newif\ifappendix
\appendixfalse
\webpagefalse

\def\authorrunning{Dominic Orchard and Nobuko Yoshida}
\def\titlerunning{Using session types as an effect system}

\author{Dominic Orchard 
\institute{Imperial College London, UK}
\and Nobuko Yoshida
\institute{Imperial College London, UK}}

\title{Using session types as an effect system}
\begin{document}
\maketitle

\begin{abstract}
  Side effects are a core part of practical programming. However, they
  are often hard to reason about, particularly in a concurrent
  setting. We propose a foundation for reasoning about concurrent side
  effects using \emph{sessions}.  Primarily, we show that 
  \emph{session types} are expressive
  enough to encode an \emph{effect system} for stateful processes. This
  is formalised via an effect-preserving encoding of a 
  simple imperative language with an effect system into the $\pi$-calculus with
  session primitives and session types (into which we encode effect specifications). 
 This result goes towards showing a connection between the expressivity
of session types and effect systems. We briefly discuss how the encoding
could be extended and applied
 to reason about and control concurrent side effects.
\end{abstract}

\let\thefootnote\relax
\footnotetext{
Update: since PLACES'15, this 
 work has been greatly expanded upon. Much of the further
work suggested in this paper is covered in our later paper 
\emph{Effects as Sessions, Sessions as Effects}
(Orchard, Yoshida) appearing in the proceedings of POPL'16.}

\ifwebpage{
\footnotetext{Author's version. To appear in the pre-proceedings of PLACES 2015.}
} 
\fi

\section{Introduction}

\emph{Side effects} such as input-output and mutation of memory are  
important features of practical programming. However, effects are
often difficult to reason about due to their implicit impact.
Reasoning about effects is even more difficult in a concurrent
setting, where interference may cause unintended non-determinism. 
For example, consider a parallel program:
$\textbf{put} \, x \, ((\textbf{get} \, x) + 2) \mid \textbf{put} \, x \, ((\textbf{get} \, x) + 1)$
 where $x$ is a mutable memory cell.
Given an initial assignment $x \mapsto 0$, the final value stored at
$x$ may be any of 3, 2, or 1 since calls to $\textbf{get}$ and $\textbf{put}$ may be interleaved. 

Many approaches to reasoning, specifying, and controlling the scope of
effects have therefore been proposed.  Seemingly orthogonally, various
approaches for specifying and reasoning about concurrent interactions
have also been developed. In this paper, we show that two particular
approaches for reasoning about effects and concurrency are in
fact \emph{non}-orthogonal; one can be embedded into the other.  We
show that \emph{session types}~\cite{YoshidaV07} for concurrent
processes are expressive enough to encode \emph{effect
  systems}~\cite{gifford1986effects,TalpinJP:typeffd} for state. 
We formalise this ability by embedding/encoding a simple imperative 
language with an effect system into the
$\pi$-calculus with session types: 
sessions simulate state and session types become effect annotations. 
Formally, our embedding maps type-and-effect judgements to session type judgements: 
\begin{equation}
\Gamma \vdash M : \tau , F \quad
\xrightarrow{\textit{embedding}} \quad
\interp{\Gamma} ; \, \mathit{res} : \, !\interp{\tau}.\textbf{end}, \effChan{} : \interp{F}
 \vdash \interp{M}
\end{equation}
That is, an expression $M$ of type $\tau$ in context $\Gamma$ performing effects $F$,  
is mapped to a process $\interp{M}$
 which sends its result over session channel $\mathit{res}$ and 
simulates effects by interactions $\interp{F}$ (defined by 
an interpretation of the effect annotation) over 
session channel $\textit{eff}$.

We start with the traditional encoding of a mutable store into the $\pi$-calculus (Section~\ref{sec:manual}) and show how its session types provide a kind of effect system. Section~\ref{sec:manual} introduces a simple imperative language, which we embed into the $\pi$-calculus with sessions (sometimes called the \emph{session calculus}) (Section~\ref{sec:embedding}).  The embedding is shown sound with respect to an equational theory for the imperative language.  Section~\ref{sec:reasoning} discusses how to extend the encoding to parallel composition in the imperative language and how the effect
information can be used to safely introduce implicit parallelism in the encoding. 

\noindent
Our embedding has been partly formalised in Agda 
and is available at {\small{\url{https://github.com/dorchard/effects-as-sessions}}}
 (Appendix~\ref{sec:agda} gives a brief description). This is used to verify the syntactic
soundness of the embedding (essentially, that types and effects are correctly translated and preserved). 

The main result of this paper is foundational and technical, about the expressive power of the $\pi$-calculus with session primitives and session types. This result has a number of possible uses:
\begin{itemize}
\item \emph{Effects systems for the $\pi$-calculus}: rather than adding an additional effect
system on top of the $\pi$-calculus, we show that existing work on session types can be reused for this purpose. 
\item \emph{Semantics of concurrency and effects}: our approach
provides an intermediate language for the semantics of effects in
a concurrent setting.
%
\item \emph{Compilation}: Related to the above, the session calculus can be used as a typed
 intermediate language for compilation, where our embedding provides 
 the translation. Section~\ref{sec:reasoning} demonstrates an optimisation 
step where safe implicit parallelism is introduced  
 based on effect information and soundness results of our embedding. 
\end{itemize}
Effect systems have been used before to reason about effects
in concurrent programs. For example,
 Deterministic Parallel Java uses an  
 effect system to check that parallel processes can
safely commute without memory races, and otherwise schedules processes
to ensure determinism~\cite{bocchino2009type}. Our approach 
allows state effects to be incorporated directly into 
concurrent protocol descriptions, reusing session types, without requiring interaction 
between two distinct systems.

\section{Simulating state with sessions}
\label{sec:manual}

\paragraph{2.1 $\,$ State via processes}

A well-known way to implement state in a process algebra is to
represent a mutable store as a server-like process (often called a
\emph{variable agent}) that offers two modes of interaction
($\mathsf{get}$ and $\mathsf{put}$).  In the
$\mathsf{get}$ mode, the agent waits to receive a value on its channel
which is then ``stored''; in the $\mathsf{put}$ it sends the stored
value. This can be implemented in the $\pi$-calculus with branching
and recursive definitions as follows (Figure~\ref{fig:pi-calc-syntax}  
describes the syntax; the calculus is based on the second system in \cite{YoshidaV07} using 
the dual channels from \cite{MostrousY15} instead of a polarity-based approach):
\begin{align}
\hspace{-0.65em}
\textbf{def} \;\, \textit{Store}(x, c) = c \rhd \{ \getS : c!\langle{x}\rangle.\textit{Store}\langle{x, c}\rangle, \; \putS : c?(y).\textit{Store}\langle{y , c}\rangle, \; \mathsf{stop} : \mathbf{0} \} \;\, \textbf{in} \;\, \textit{Store}\langle{i, \effChan{}}\rangle 
\label{def:store}
\end{align}
where \textit{Store} is parameterised by the stored value $x$ and a 
session channel $c$. 
That is, \emph{Store} provides a choice (by $\rhd$) over channel $c$ between three 
behaviours labelled $\getS$, $\putS$, and $\mathsf{stop}$. The $\getS$ branch sends the
state $x$ on $c$ and then recurses with the same parameters, preserving
the stored value.  The $\putS$ branch receives $y$ which then becomes the
new state by continuing with recursive call $\textit{Store}\langle{y, c}\rangle$. 
The $\mathsf{stop}$ branch provides finite interaction by
terminating the agent. 
The store agent is initialised with a value $i$ and channel $\effChan{}$.

The following parameterised operations \textit{get} and \textit{put} 
then provide interaction with the store:  
\begin{align}
\mathit{get}(c)(x).P = \overline{c} \lhd \getS \, . \, \overline{c} ?(x).P  
\qquad\quad
\mathit{put}(c)\langle{V}\rangle.P = \overline{c} \lhd \putS \, . \, \overline{c}!\langle{V}\rangle.P 
\label{def:get-put}
\end{align}
where $\overline{c}$ is the opposite endpoint of a channel, 
and \textit{get} selects (by the $\lhd$ operator) the $\getS$ branch
then receives a value which is bound to $x$ in the scope of $P$. The
\textit{put} operation selects its relevant branch then sends a value $V$ before 
continuing as $P$.

A process can then use \textit{get} and \textit{put} for stateful computation
by parallel composition with \textit{Store}, \eg{} $\textit{get}(\effChan)(
x).\textit{put}(\effChan)\langle{x+1}\rangle.\textbf{end} \; | \; \textit{Store}\langle{i, \effChan}\rangle$ increments the initial value. 

\begin{figure}[t]
\begin{minipage}{1\linewidth}
\begin{align*}
\begin{array}{c}
\hspace{-1.1em}
(\textit{value variables}) \quad v ::= x, y, z 
\quad\quad
(\textit{session channel variables}) \quad c, d, \overline{c}, \overline{d} \\[0.2em]
\begin{array}{rrlll}
(\textit{values}) & V ::=   
& \!\! C \mid v         & & \textit{constants / variables} \\[0.2em]
(\textit{processes}) & P,Q  ::= 
 & \!\!  c ? (x) . P & \hspace{-0.5em} \mid \;  c ! \langle{V}\rangle.P \hspace{1.5em} & \textit{receive / send} \\
& \mid & \!\!  c ? (d).P & \hspace{-0.5em} \mid \; c ! \langle{d}\rangle.P & \textit{channel receive / send} \\
& \mid  & \!\!  c \rhd \{\tilde{l} : \tilde{P}\} & \hspace{-0.5em} \mid \; c \lhd l . P & \textit{branching / selection} \\
& \mid  & \!\!\! \textbf{def} \; X(\tilde{x},\tilde{c}) = P \; \textbf{in} \; Q & 
\hspace{-0.5em} \mid \; X\langle{\tilde{V}, \tilde{c}}\rangle & \textit{recursive definition / use} \\
& \mid & \!\!\! \nu c . P             & & \textit{channel restriction} \\
& \mid & \!\!\! (P \, | \, Q)              & & \textit{parallel composition} \\
& \mid  & \!\!\! \mathbf{0}            & & \textit{nil process} \\[0.2em]
\end{array} \\[0.2em]
(\textit{value-types}) \quad  \tau ::= \textbf{unit} \mid \textbf{nat} \mid S  
\qquad
(\textit{contexts}) \quad \Gamma ::= \emptyset \mid \Gamma, x : \tau \mid \Gamma, X : (\tilde{\tau}, \tilde{S}) 
\\[-0.5em]
\end{array}
\end{align*}
($l$ ranges over labels, $\tilde{l} : \tilde{P}$ over sequences of label-process pairs,
and $\tilde{e}$ over syntax sequences)
\end{minipage}
\caption{Syntax of $\pi$-calculus with recursion and sessions}
\label{fig:pi-calc-syntax}
\end{figure}
 
\paragraph{2.2 $\,$ Session types}

Session types provide descriptions (and restrictions) of 
the interactions that take place over channels~\cite{YoshidaV07}.
%
Session types record sequences of typed \emph{send} ($![\tau]$) and \emph{receive} ($?[\tau]$) interactions, terminated by the \textbf{end} marker, branched by \emph{select} ($\oplus$) and \emph{choice} ($\&$) interactions, with cycles provided by a fixed point $\mu \alpha$ and session variables $\alpha$: 
\begin{align*}
S,T ::= & \; ![\tau] . S \mid \; ?[\tau] . S \mid \oplus [l_1 : S_1, \ldots, l_n : S_n] \mid \&[l_1 : S_1, \ldots, l_n : S_n] \mid \mu \alpha . S \mid \alpha \mid \mathbf{end}
\end{align*}
\noindent
where $\tau$ ranges over value types $\textbf{nat}, \textbf{unit}$ and session channels $S$,
and $l$ ranges over labels. 

Figure~\ref{fig:pi-calc-typing} (p.~\pageref{fig:pi-calc-typing}) gives the rules of the session typing system (based 
on that in~\cite{YoshidaV07}). Session typing judgements for processes have the form $\Gamma; \Delta \vdash P $ meaning a process $P$ has value variables $\Gamma = x_1 : \tau_1 \ldots x_n : \tau_n$ and session-typed channels $\Delta = c_1 : S_1, \ldots c_n : S_n$. 

For a session $S$, its \emph{dual}
$\overline{S}$ is defined in the usual way~\cite{YoshidaV07}:
\begin{align*}
\begin{array}{c}
\du{![\tau].S} = ?[\tau].\du{S} \quad
\du{?[\tau].S} = ![\tau].\du{S} \quad
\du{\mu \alpha . S} = \mu \alpha . \du{S} \quad
\du{\alpha} = \alpha \quad
\du{\mathbf{end}} = \mathbf{end}
\\
\du{\oplus[l_1 : S_1, \ldots, l_n : S_n]} = \&[l_1 : \du{S_1}, \ldots, l_n : \du{S_n}] \quad
\du{\&[l_1 : S_1, \ldots, l_n : S_n]} = \oplus[l_1 : \du{S_1}, \ldots, l_n : \du{S_n}]
\end{array}
\end{align*}
For some state type $\tau$ and initial value $i : \tau$, 
the \textit{Store} process \eqref{def:store} has session judgement\ifappendix 
 (for which the derivation is shown in Appendix~\ref{sec:store-types}):
\else
:
\fi
\begin{align*}
\Gamma; \effChan{} : \mu \alpha . \, \& [\getS{} : \, ! [\tau] . \alpha, \putS{} : ? [\tau] . \alpha, \stopS{} : \mathbf{end}] 
\vdash \mathit{Store}\langle{i, \effChan{}}\rangle
\end{align*}
That is, $\effChan{}$ is a channel over which there is an sequence of offered
choice between the $\getS{}$ branch, which sends a value, $\putS$ branch
which receives a value, and is terminated by the $\stopS{}$ branch. 

The session judgements for \textit{get} and \textit{put} \eqref{def:get-put} are then:
\begin{align}
\dfrac{\Gamma, x : \tau; \, \Delta, \overline{\effChan} : S \vdash P}
{\Gamma; \Delta , \overline{\effChan} : \oplus [\getS : \, ? [\tau] . S]
\vdash \mathit{get}(\effChan)(x).P}
\quad\;\,
&
\quad\;\,
\dfrac{\Gamma; \emptyset \vdash V : \tau \quad \Gamma; \Delta, \overline{\effChan} : S \vdash P}
{\Gamma; \Delta , \overline{\effChan} : \oplus [\putS : \, ! [\tau] . S]
\vdash \mathit{put}(\effChan)\langle{V}\rangle.P}
\label{eq:get-put}
\end{align}
We use a variant of session typing where selection terms $\rhd$, used by
$\mathit{get}$ and $\mathit{put}$, have a selection session type $\oplus$ 
with only the selected label (seen above), and not the full range
of labels offered by its dual branching process, which would be 
$\oplus[\getS : ?[\tau].S, \putS : ![\tau].S, \stopS : \mathbf{end}]$ for both. Duality of
select and branch types is achieved by using session subtyping
to extend select types with extra labels~\cite{CDY2014}
 (see Appendix~\ref{subtyping-discussion} for details). 

Throughout we used the usual convention of eliding a trailing
$\mathbf{0}$, \eg{}, writing $r!\langle{x}\rangle$ instead of
$r!\langle{x}\rangle.\mathbf{0}$, and likewise for session types,
\eg{}, $![\tau]$ instead of $![\tau].\mathbf{end}$. 

\begin{figure}[t]
\framebox{
\begin{minipage}{0.94\linewidth}
\begin{align*}
\begin{array}{c}
\framebox{$\Gamma; \Delta \vdash V : \tau$} 
\;\;
(\textit{value typing}) \;\;
\;
\trule{const}
\dfrac{C : C_\tau}{\Gamma; \emptyset \vdash C : C_\tau} 
\;\,
\trule{var}
\dfrac{v : \tau \in \Gamma}{\Gamma; \emptyset \vdash v : \tau}
\;\,
\trule{suc}
\dfrac{\Gamma; \emptyset \vdash V : \textbf{nat}}{\Gamma; \emptyset \vdash \mathbf{suc} \, V : \textbf{nat}}
\!\!\!\!
\\[1.8em]
\hspace{-2em}
\framebox{$\Gamma; \Delta \vdash P$} 
\quad 
(\textit{process typing}) \\[1.2em]
\trule{end}\;
{\Gamma; \tilde{c} : \textbf{end} \vdash \textbf{0}}
\quad
\trule{par}
\dfrac{\Gamma; \Delta_1 \vdash P \quad \Gamma; \Delta_2 \vdash Q} 
      {\Gamma; \Delta_1,\Delta_2 \vdash P \mid Q} 
\quad
\trule{restrict}
\dfrac{\Gamma; \Delta, c : S, \overline{c} : \overline{S} \vdash P}
      {\Gamma; \Delta \vdash \nu c . P}
\\[1.5em]
\begin{array}{cc}
\trule{def}
\inference{\Gamma, X : (\tilde{\tau}, \tilde{S}), \tilde{x} : \tilde{\tau} ; \; \tilde{c} : \tilde{S} \vdash P \\ \Gamma, X : (\tilde{\tau}, \tilde{S}); \Delta \vdash Q} 
      {\Gamma; \Delta \vdash \textbf{def} \; X(\tilde{x}, \tilde{c}) = P \; \textbf{in} \; Q}
&
\trule{dvar}
\dfrac
      {\Gamma; \emptyset \vdash \tilde{V} : \tilde{\tau}}
      {\Gamma, X : (\tilde{\tau}, \tilde{S}); \tilde{c} : \tilde{S}, \tilde{d} : \textbf{end} \vdash X\langle{\tilde{V}, \tilde{c}}\rangle}
\end{array}\\[2.2em]
\begin{array}{cc}
\trule{chan-recv}
 \dfrac{\Gamma; \Delta, c : T, d : S \vdash P}{\Gamma; \Delta, c : ?[S].T \vdash c?(d).P}
&
\trule{chan-send}
 \dfrac{\Gamma; \Delta, c : T \vdash P}{\Gamma; \Delta, c : \, ![S].T, d : S \vdash c!\langle{d}\rangle.P}
\\[2em]
\end{array} \\
\begin{array}{cc}
\trule{recv}
\dfrac{\Gamma, x : \tau; \Delta, c : S \vdash P}
      {\Gamma; \Delta, c : ?[\tau]. S  \vdash c?(x).P}
&
\trule{send}
\dfrac{\Gamma; \emptyset \vdash V : \tau \quad \Gamma ; \Delta , c : S \vdash P}
      {\Gamma; \Delta, c : \, ![\tau]. S  \vdash c!\langle{V}\rangle.P}
\end{array}\\[2em]
\begin{array}{cc}
\trule{branch}
\dfrac{\Gamma; \Delta, c : S_i \vdash P_i}
      {\Gamma; \Delta, c : \& [\tilde{l} : \tilde{S}] \vdash c \rhd \{\tilde{l} : \tilde{P}\}}&
\trule{select}
\dfrac{\Gamma; \Delta, c : S\ \vdash P}
      {\Gamma; \Delta, c : \oplus [l : S] \vdash c \lhd l . P}
\end{array}
\end{array}
\end{align*}
where $\tilde{x} : \tilde{\tau}$ is shorthand for a sequence of variable-type pairs, 
and similarly $\tilde{c} : \tilde{S}$ for channels, $\tilde{l} : \tilde{S}$ for labels and sessions, and $\tilde{V}$ for a sequence of values.
\end{minipage}
}
\caption{Session typing relation over the $\pi$-calculus 
with recursion and sessions~\cite{YoshidaV07}.} 
\label{fig:pi-calc-typing}
\end{figure} 

\paragraph{2.3 $\,$ Effect systems}

Effect systems are a class of static analyses for 
effects, such as state or exceptions~\cite{gifford1986effects,nielson1999type,TalpinJP:typeffd}. 
Traditionally, effect systems are described as syntax-directed analyses 
by augmenting typing rules with effect judgements, \ie{},
$\Gamma \vdash M : \tau, F$ where $F$ describes the effects of $M$ --
usually a set of effect tokens (but often generalised, \eg{}, in~\cite{nielson1999type}, to arbitrary semi-lattices, monoids, or semirings). 

We define the \emph{effect calculus}, a simple imperative language with effectful operations and a type-and-effect system defined in terms of an abstract monoidal effect
algebra. 
%
Terms comprise variables, \letb-binding, operations, and 
constants, and types comprise value types for natural numbers and unit: 
\begin{align*}
M, N ::= x \mid \synLet{M}{x}{N} \mid \mathit{op} \, M \mid c 
\qquad \qquad
\tau, \sigma ::= \unit \mid \nat
\end{align*}
where $x$ ranges over variables, $\textit{op}$ over unary operations, and
$c$ over constants. We do not include 
function types as there is no abstraction (higher-order calculi are discussed
in Section~\ref{sec:higher-order-discussion}). 
Constants and operations can be effectful and are instantiated to provide 
application-specific effectful operations in the calculus. As defaults, 
we include zero and unit constants $0, \unitO \in c$ and a pure successor operation
for natural numbers $\sucO \in \textit{op}$. 

\begin{definition}[Effect system]
  Let $\mathcal{F}$ be a set of effect annotations with a 
  monoid structure $(\mathcal{F}, \bullet, I)$ where $\bullet$ combines effects
  (corresponding to sequential composition) and $I$ is the trivial effect 
  (for pure computation). Throughout $F, G, H$
  will range over effect annotations.  
%
%

Figure~\ref{fig:effect-calculus} defines the
type-and-effect relation. The (var) 
rule marks variable use as pure (with $I$). In
(let), the left-to-right evaluation order of \letb{}-binding is exposed by
the composition order of the effect $F$ of the bound term $M$ followed
by effect $G$ of the \letb{}-body $N$. The (const) rule introduces a constant of type $C_\tau$  
with effects $C_F$, and (op) applies an operation to its pure argument of type $\textit{Op}_\sigma$, returning a result of type $\textit{Op}_\tau$ with effect $\textit{Op}_F$. 

\begin{figure}[t]
\begin{center}
{{ 
\begin{minipage}{1\linewidth}
\begin{align*}
\hspace{-0.3em}
\begin{array}{c}
\rulec{var}\dfrac{x : \tau \in \Gamma}{\Gamma \vdash x : \tau , I}
\;\;\;\;
\rulec{let}\dfrac{\Gamma \vdash M : \sigma, F \quad
                  \Gamma, x : \sigma \vdash N : \tau, G}
                 {\Gamma \vdash \synLet{M}{x}{N} : \tau, F \bullet G}
\;\;\;\;
\rulec{const}\dfrac{}{\Gamma \vdash c : C_\tau , C_F} 
\;\;\;\;
\rulec{op}\dfrac{\Gamma \vdash M : \textit{Op}_\sigma, I}{\Gamma \vdash \textit{op} \; M : \textit{Op}_\tau , \textit{Op}_F}
\end{array}
\end{align*}
\end{minipage}
}}
\caption{Type-and-effect system for the effect calculus}
\label{fig:effect-calculus}
\vspace{-0.7em}
\end{center}
\end{figure}
\end{definition}

\paragraph{2.4 $\,$ State effects}
The effect calculus can be instantiated with different notions of effect. 
For state, we use the effect monoid $(\mathsf{List} \, \{\getF \,
\tau, \putF \, \tau\}, ++, [\,])$ of lists of effect tokens, where
$\getF$ and $\putF$ represent \emph{get} and \emph{put} effects
parameterised by a type $\tau$, $++$ concatenates lists and $[\,]$ is the empty list.
Many early effect systems annotated terms with sets of effects. 
Here we use lists to give a more precise account of state which includes the 
order in which effects occur. This is often described as a \emph{causal} effect
system. 

Terms are extended with 
constant $\getO$ and unary operation $\putO$ where
$\emptyset \vdash \getO : \tau , [\getF \, \tau]$
and
$\Gamma \vdash \putO \; M : \unit , [\putF \, \tau]$ for $\Gamma \vdash M : \tau, I$. 
For example, the following is a valid judgement: 
\begin{equation}
\emptyset \vdash 
\synLet{\getO\,}{x}{\,\putO \, (\sucO \; x)} \, : \nat , [\getF \, \nat, \putF \, \nat]
\label{eq:eff-sample}
\end{equation}
Type safety of the store is enforced by requiring that any \textit{get} effects must
have the same type as their nearest preceding \textit{put} effect. We implicitly apply this condition throughout. 

\paragraph{2.5 $\,$ Sessions as effects}
%
The session types of processes 
interacting with \textit{Store} provide the same information as the state effect system. 
Indeed, we can define a bijection between state effect annotations  
and the session types of $\textit{get}$ and $\textit{put}$ \eqref{eq:get-put}: 
\begin{align}
\interp{[\,]} = \textbf{end} 
\qquad
\interp{(\getF \, \tau) :: F} = \oplus [\getS : ?[\tau] . \interp{F}] 
\qquad
\interp{(\putF \, \tau) :: F} = \oplus [\putS : \, ![\tau] . \interp{F}]
\label{def:effect-translation}
\end{align}
where $::$ is the \emph{cons} operator for lists. 
Thus processes interacting with \textit{Store} have 
 session types corresponding to effect annotations. For example, the following
has the same state semantics as \eqref{eq:eff-sample} and isomorphic session types: 
\begin{equation}
\emptyset; \overline{\effChan} : \interp{[\getF \, \nat,  \putF \, \nat]} \vdash \textit{get}(\effChan{})(x).\textit{put}(\effChan{})\langle{\mathbf{suc} \, x}\rangle
\end{equation}

\section{Embedding the effect calculus into the $\pi$-calculus}
\label{sec:embedding}

Our embedding is based on the embedding of the call-by-value
$\lambda$-calculus (without effects) into the $\pi$-calculus~\cite{MilnerR:funp,SangiorgiD:picatomp}
 taking $\synLet{M}{x}{N} =
(\lambda x . N) \, M$. Since effect calculus terms return a result
and $\pi$-calculus processes do not, the embedding is parameterised by a \emph{result
  channel} $r$ over which the return value is sent, written $\interp{-}_r$. 
Variables and pure \letb-binding are embedded: 
\begin{align}
\interp{\synLet{M}{x}{N}}_r = \; \nu q . \;\, (\interp{M}_q \mid \overline{q}?(x). \interp{N}_r)\hspace{4em}
\interp{x}_r                 = \; r!\langle{x}\rangle
\label{eq:pure-let} 
\end{align}
Variables are simply sent over the result channel.  
For \letb, an intermediate channel $q$ is created over which the result of 
the bound term $M$ is sent by the left-hand parallel process $\interp{M}_q$ and received and bound to $x$ by the right-hand process before continuing with $\interp{N}_r$. This enforces a left-to-right, CBV evaluation order (despite the parallel composition). 

Pure constants and unary operations can be embedded 
similarly to variables and \letb{} given suitable value operations in the
 $\pi$-calculus. For example, successor and zero are embedded as:
\begin{align}
\interp{\sucO \, M}_r = \; \nu q . \;\, (\interp{M}_q \mid \, \overline{q}?(x) . r!\langle{\textbf{suc} \; x}\rangle)
\hspace{4em}
\interp{\mathsf{zero}}_r = \, r!\langle{\mathbf{zero}}\rangle
\label{eq:pure-op}
\end{align}
Given a mapping $\interp{-}$ from effect calculus types to corresponding value types in
the $\pi$-calculus, the above embedding of terms \eqref{eq:pure-let},\eqref{eq:pure-op} 
can be extended to typing judgements as follows (where $\interp{\Gamma}$ interprets the type of each free-variable assumption pointwise, preserving the structure of $\Gamma$):
\begin{align}
\interp{\Gamma \vdash M : \tau}_r \, = \; \interp{\Gamma}; r : \, !\interp{\tau}.\textbf{end} \vdash
\interp{M}_r 
\end{align}

\paragraph{With effects}
Our approach to embedding effectful computations is to simulate
effects by interacting with an effect-handling agent over a session
channel. The embedding, written $\interp{-}_r^{\effChan}$, maps a 
judgement $\Gamma \vdash M : \tau, \eff{F}$ to a session type
judgement with channels $\Delta = (r :
\,!\interp{\tau}.\textbf{end}, \, {\effChan} : \interp{F})$, \ie{}, 
the effect annotation $F$ is interpreted as the session type of 
channel ${\effChan}$. For state, this interpretation 
is defined as in eq. \eqref{def:effect-translation}. 
The embedding first requires an intermediate step, written $\interpA{-}^{\ei,\eo}_r$
\begin{align}
\interpA{\Gamma \vdash M : \tau , F}^{\ei,\eo}_r \, & = \;
\forall g . \quad  \interp{\Gamma}; \;
r : \, !\interp{\tau}, \, 
\ei : \, ? \interp{F \bullet g}, \, \overline{\eo} : ! \interp{ g } \;  
\vdash \interpA{M}^{\ei,\eo}_r 
\end{align}
where $\ei$ and $\eo$ are channels over which channels for simulating 
effects (which we call \emph{effect channels}) are communicated: $\ei$ receives an effect channel of session type
$\interp{F \bullet g}$ (\ie{}, capable of carrying out effects $F
\bullet g$) and $\overline{\eo}$ sends a channel of session type $\interp{g}$
(capable of carrying out effects $g$). Here the effect $g$ is 
universally quantified at the meta level. This provides a way to
``thread'' a channel for effect interactions through a computation, such
as in the case of \letb-binding (see below).  

The embedding of effect calculus terms is then defined:
\begin{align}
\notag \interpA{\synLet{M}{x}{N}}^{\ei,\eo}_r & = \nu q, \ea \, . \; (\interpA{M}^{\ei,\ea}_q\mid \overline{q}?(x). \interpA{N}^{\ea,\eo}_r)\\[-0.1em]
\notag \interpA{x}^{\ei,\eo}_r                & = \ei?(c).r!\langle{x}\rangle.\overline{\eo}!\langle{c}\rangle \\[-0.15em]
\notag\interpA{C}^{\ei,\eo}_r                & = \ei?(c).\interp{C}_r.\overline{\eo}!\langle{c}\rangle \hspace{4.37em} \textit{(when $C$ is pure)} \\[-0.15em]
\label{def:embedding} 
\interpA{op \, M}^{\ei,\eo}_r          & =  \ei?(c).\interp{op \, M}_r.\overline{\eo}!\langle{c}\rangle \hspace{2.9em} \textit{(when $op$ is pure)}
\end{align}
The embedding of variables is straightforward, where an effect channel
$c$ is received on $\ei$ and then sent without use 
on $\overline{\eo}$. Embedding pure operations and constants is similar, reusing the 
pure embedding defined above in equation \eqref{eq:pure-op}.

The \letb{} case resembles the pure embedding of \letb{}, but threads through an effect channel to each sub-expression. An intermediate channel $\ea$ is introduced over which an 
effect channel is passed from the embedding of
$M$ to $N$. Let $\Gamma \vdash M : \sigma, F$ and 
$\Gamma, x : \sigma \vdash N : \tau, G$ then 
in the embedding of $\synLet{M}{x}{N}$ 
the universally quantified effect variable $\forall g$ for $\interpA{M}^{\ei,\ea}_q$ 
 is instantiated to $G \bullet h$. The following partial session-type derivation for the $\letb$ encoding shows the propagation of effects via session types: 

{\small{
\begin{equation*}
\dfrac{
\dfrac{
\dfrac{
q : !\interp{\sigma}, 
\ei : ?\interp{F \bullet g}, 
\overline{\ea} : !\interp{g}
\vdash \interpA{M}^{\ei,\ea}_q
\;\;\;
g \mapsto G \bullet h}{
  q : !\interp{\sigma}, 
\ei : ?\interp{F \bullet (G \bullet h)}, 
\overline{\ea} : !\interp{G \bullet h}
\vdash \interpA{M}^{\ei,\ea}_q}
\qquad 
{ \overline{q} : ?\interp{\sigma}, r : !\interp{\tau}, 
  \ea : ?\interp{G \bullet h}, 
  \overline{\eo} : !\interp{h}
  \vdash \overline{q}?(x). \interpA{N}^{\ea,\eo}_r}}
  {
 r : !\interp{\tau},\ei : ?\interp{F \bullet (G \bullet h)}, \overline{\eo} :!\interp{h}, q : !\interp{\sigma}, \overline{q} : ?\interp{\sigma}, 
  \overline{\ea} :!\interp{G \bullet h}, \ea : ?\interp{G \bullet h}
  \vdash \interpA{M}^{\ei,\ea}_q\mid \overline{q}?(x). \interpA{N}^{\ea,\eo}_r
  }
}
{
 r : !\interp{\tau}, \ei : ?\interp{(F \bullet G) \bullet h}, \overline{\eo} : !\interp{h} 
\, \vdash \, 
\nu q, \ea . \; (\interpA{M}^{\ei,\ea}_q\mid \overline{q}?(x). \interpA{N}^{\ea,\eo}_r)
}
\end{equation*}
}}
Associativity of $\bullet$ is used in the last rule to give the correct typing
for the embedding of $\mathbf{let}$. 

The \emph{get} and \emph{put} operations of our state effects are embedded similarly
to equation \eqref{def:get-put} (page \pageref{def:get-put}), but with the receiving
and sending of the effect channel which is used to interact with the store: 
\begin{align}
\notag \interpA{\getO{}}^{\ei,\eo}_r & = \ei?(c) . {c} \lhd \getS \, . \, {c} ?(x).r!\langle{x}\rangle.\overline{\eo}!\langle{c}\rangle \\
\interpA{\putO{} \, M}^{\ei,\eo}_r & =
\nu q . \; (\interp{M}_{q} \mid \; \ei?(c).\overline{q}?(x).{c} \lhd \putS \, . \, {c} !\langle{x}\rangle.r!\langle{\mathbf{unit}}\rangle.\overline{\eo}!\langle{c}\rangle)
\end{align}
The embedding of $\getO{}$ receives channel $c$ over which
it performs its effect by selecting the $\mathsf{get}$ branch and receiving 
 $x$ which is sent as the result on $r$ before sending $c$ on $\overline{\eo}$. 
The $\putO$ embedding is similar to $\getO{}$ and \letb{}, 
but using the pure embedding $\interp{M}_q$ since $M$ is pure. Again,
the effect channel $c$ is received on $\ei$ and is used to interact
with the store (the usual $\mathsf{put}$ action) before being sent on
$\du{\eo}$. 

The full embedding is then defined in terms of the intermediate embedding
$\interpA{-}$ as follows:
\begin{align}
\!\!\!\!\! \interp{\Gamma \vdash M : \tau , F}^{\effChan}_r = 
\interp{\Gamma}; \,  r : \, !\interp{\tau}, {\mathit{eff}} : \interp{F}
\, \vdash \, 
\nu \ei,\eo . \; (\interpA{\Gamma \vdash M : \tau , F}^{\ei,\eo}_r \mid \, \overline{\ei}!\langle{{\effChan}}\rangle.\eo?(c))
\label{def:top-level}
\end{align}
where ${\effChan{}}$ is the free session channel over which effects are performed. 
Note that $c$ (received on $\eo$) is never used and thus has session type 
$\mathbf{end}$. 

Finally, the embedded program is composed in parallel with the variable
agent, for example:
\begin{align}
\textbf{def} \; \mathit{Store}(x, c)  = \ldots \text{(see eq.~\eqref{def:store})} \; \textbf{in} \; \mathit{Store}\langle{0, \overline{\effChan}}\rangle \mid \interp{\synLet{\mathsf{get}}{x}{\mathsf{put}\, (\sucO \, x)}}^{\effChan}_r
\end{align}

\paragraph{3.1 $\,$ Soundness}
\label{sec:soundness}

The effect calculus exhibits the equational theory defined by the
relation $\equiv$ in Figure~\ref{fig:equations}, 
which enforces monoidal properties on effects and the effect algebra
(assoc),(unitL),(unitR), and which allows pure computations to commute 
with effectful ones (comm). 
Our embedding is sound with respect to these equations 
up to weak bisimulation of session calculus processes (see, \eg{}~\cite{KOUZAPAS_2014}, for more on the weak bisimulation relation). 

\begin{theorem}[Soundness]
If
$\Gamma \vdash M \equiv N : \tau, F$ 
then
$\interp{\Gamma}; 
(r : !\interp{\tau}.\mathbf{end}, \textit{e} : \interp{F})
\vdash \interp{M}^{e}_r \approx \interp{N}^{e}_r$
\end{theorem}


%
%

\noindent
Appendix~\ref{sec:proofs} gives the proof. 
Proof of soundness with respect to (comm) requires an additional restriction on the effect algebra, that:
\begin{equation}
\forall F, G . \;\;\;\, (F \bullet G) \equiv I \;\; \Rightarrow \;\, (F \equiv G \equiv I)
\end{equation}
That is, if the composition of two effects is pure,
then both components are themselves pure. This is equivalent to requiring 
that there are no inverse elements (with respect to $\bullet$) in the
effect set $\mathcal{F}$.  The state effect system described here
satisfies this additional effect-algebra condition since any
two lists whose concatenation is the empty list implies that both
lists are themselves empty.

The soundness proof for (comm) is split into an additional lemma that shows
the encoding of effectful terms marked as pure with $I$ can be factored through the encoding of pure terms 
(shown at the start of Section~\ref{sec:embedding}). That is, in an appropriate session calculus context $C$, then:
\begin{equation*}
C[\interpA{\Gamma \vdash M : \tau, I}^{\ei,\eo}_r] \; \approx \; C[\ei?(c).\overline{\eo}!\langle{c}\rangle \mid \interp{M}_r]
\end{equation*}
Thus, the intermediate encoding of a pure term is weakly bisimilar to a pure
encoding (without effect simulation) composed in parallel with a
process which receives an effect-simulating channel $c$ on $\ei$ and
sends it on $\overline{\eo}$ without any use.
Appendix~\ref{sec:proofs} provides the details and proof.

\begin{figure}[t]
\vspace{-0.25em}
{\small{
\begin{align*}
\begin{array}{c}
\trule{assoc}
\dfrac{\Gamma \vdash M : \sigma , F \quad\; \Gamma, x : \sigma \vdash N : \tau', G
                                    \quad\; \Gamma, y : \tau' \vdash P : \tau, H
\quad\; x \not\in FV(P)}
      {(\synLet{(\synLet{M}{x}{N})}{y}{P})
        \equiv (\synLet{M}{x}{(\synLet{N}{y}{P}))} : \tau , F \bullet G \bullet H}
\\[1.3em]
\trule{unitL}
\dfrac{\Gamma \vdash x : \sigma , I \quad \Gamma, y : \sigma \vdash M : \tau , F} 
      {\Gamma \vdash (\synLet{x}{y}{M}) \equiv M[x/y] : \tau, F}
\quad
\trule{unitR}
\dfrac{\Gamma \vdash M : \tau , F}
      {\Gamma \vdash (\synLet{M}{x}{x}) \equiv M : \tau , F}
 \\[1.3em]
\trule{comm}
\dfrac{\Gamma \vdash M : \tau_1, \eff{I} \qquad \Gamma \vdash N : \tau_2, \eff{F} \qquad \Gamma, x : \tau_1 , y: \tau_2 \vdash P : \tau, \eff{G} \quad x \not\in FV(N) \quad y \not\in FV(M)}
      {\Gamma \vdash \synLet{M}{x}{(\synLet{N}{y}{P})}\equiv
       \synLet{N}{y}{(\synLet{M}{x}{P})}: \tau, \eff{F \bullet G}}
\end{array}
\label{gen-alg}
\end{align*}}}
\vspace{-1em}
\caption{Equations of the effect calculus}
\label{fig:equations}
\vspace{-0.3em}
\end{figure}

\section{Discussion}

\paragraph{Concurrent effects}
\label{sec:reasoning}

In a concurrent setting, side effects can lead to non-determinism and race conditions.  For example, the program $\textbf{put} \, (\textbf{get}\, + 2) \!\mid\! \textbf{put} \, (\textbf{get} \, + 1)$ has three possible final values for the store due to arbitrarily interleaved $\textbf{get}$ and $\textbf{put}$ operations.

Consider an extension to the source language which adds an operator
for parallel composition $\texttt{|}$ (we elide details of the type-and-effect rule, but an
 additional effect operator describing parallel effects can be included). 
We might then attempt the following encoding, composing encodings of sub-terms
in parallel: 
\begin{equation*}
\interp{M \, \texttt{|} \, N}^{\effChan}_r 
 = \nu q_1, q_2 \, . \, (\interp{M}^{\effChan}_{q_1} \mid \interp{N}^{\effChan}_{q_2} \mid \overline{q_1}?(x).\overline{q_2}?(y).r!\langle{(x,y)}\rangle)
\end{equation*}
where $q_1$ and $q_2$ are the result channels for each term, from which
the results are paired and sent over $r$. This encoding is not well-typed 
under the session typing scheme: the (par) rule (see Figure~\ref{fig:pi-calc-typing}, 
p.~\pageref{fig:pi-calc-typing}) 
requires that the session channel environments of each process be disjoint, but
$\effChan$ appears on both sides. 
Thus, session types naturally prevent effect interference. 

Concurrent programs with effects can be encoded by extending our
session calculus with \emph{shared channels}, which can be used in parallel
and over which session channels are initiated~\cite{YoshidaV07}.  Shared channels can be used in the encoding of effectful operations ($\textbf{get}$/$\textbf{put}$)
 to lock the store, providing atomicity of
each effectful interaction via the following redefinitions:
%
\begin{align*}
\begin{array}{c}
\textbf{def} \; \textit{Store}(x, k) = \textbf{accept} \; k(c).c \rhd \{ \getS : c!\langle{x}\rangle.\textit{Store}\langle{x, k}\rangle, \; \putS : c?(y).\textit{Store}\langle{y , k}\rangle, \; \mathsf{stop} : \mathbf{0} \} \; \textbf{in} \; \textit{Store}\langle{i, k}\rangle \\[0.2em]
\mathit{get}(k)(x).P = \mathbf{request} \; k(c).\overline{c} \lhd \getS \, . \, \overline{c} ?(x).P  
\qquad\quad
\mathit{put}(k)\langle{V}\rangle.P = \mathbf{request} \; k(c).\overline{c} \lhd \putS \, . \, \overline{c}!\langle{V}\rangle.P 
\end{array}
\end{align*}
where $k$ is a shared channel and \textbf{request}/\textbf{accept} initiate separate binary sessions between the store process and the effectful operations. This ensures atomicity of each 
 side-effect interaction (selecting then send/receiving).
Note however that the type of $k$ will describe the behaviour of all
possible initiated sessions. Via subtyping, this becomes $\oplus [\mathsf{get} : ?[\tau],
\mathsf{put} : ![\tau]]$ which conveys no information about which effect operations
occur, nor their ordering. This reflects the non-deterministic behaviour of concurrent effects, where each effectful operation (as an atom) could be arbitrarily interleaved. 

Further work is to study various other kinds of concurrent effect interaction that could be described 
using the rich language of the session calculus and variations of our embedding. 

%
%
%
%
%


\paragraph{Compiling to the session calculus}

One use for our embedding is as a typed intermediate language for a compiler since the $\pi$-calculus with session primitives provides an expressive language for concurrency.  For example, even without explicit concurrency in the source language our encoding can be used to introduce implicit parallelism as part of a compilation step via the session calculus. In the case of 
compiling a term which matches either side of the (comm) rule above, a pure term $M$ can be computed in parallel with $N$, \ie{}, given terms $\Gamma \vdash M : \tau_1, \eff{I}$ and $\Gamma \vdash N : \tau_2, \eff{F}$ and $\Gamma, x : \tau_1 , y: \tau_2 \vdash P : \tau, \eff{G}$ where $x \not\in FV(N),
 y \not\in FV(M)$ then the following specialised encoding can be given: 
\begin{align*}
\interpA{\synLet{N}{y}{(\synLet{M}{x}{P})}}^{\ei,\eo}_r = & \\[-0.2em]
\interpA{\synLet{M}{x}{(\synLet{N}{y}{P})}}^{\ei,\eo}_r = & \; 
\nu \, q, s, \ea . \, (\interp{M}_q \mid \interpA{N}^{\ei,\ea}_{s} \mid \overline{q}?(x).\overline{s}?(y). \interpA{P}^{\ea,\eo}_r) 
\end{align*}
This alternate encoding introduces the opportunity for parallel evaluation of $M$
and $N$. It is enabled by the effect system (which annotates $M$ with $I$) and it is sound: it is weakly bisimilar to the usual encoding (which follows from the soundness proof of (comm) in Appendix~\ref{sec:proofs}). 

\section{Summary and further work}
\label{sec:higher-order-discussion}

This paper showed that sessions and session types are expressive
enough to encode stateful computations with an effect system. We
formalised this via a sound embedding of a simple, and general, effect
calculus into the session calculus.
 Whilst we have focussed on causal state effects, our
effect calculus and embedding can also be instantiated for I/O effects, 
where $\textit{input}$/$\textit{output}$ operations and effects have a similar 
form to $\textit{get}$/$\textit{put}$. We considered only state effects on 
 a single store, but traditional effect systems account for multiple
stores via \emph{regions} and first-class reference values.
 Our approach could be extended  
with a store and session channel per region or reference.  This is further work. 
Other instantiations of our effect calculus/embedding are further work, for example, for set-based effect systems.

Effect reasoning is more difficult in higher-order settings as the effects
of abstracted computations are locally unknown.  Effect systems
account for this by annotating function types with the \emph{latent
  effects} of a function which are delayed till application. 
%
%
A potential encoding of such types into session types is: 
\begin{align*}
\interp{\sigma \xrightarrow{F} \tau} 
= \, !\interp{\sigma} \, . \, ![?\interp{F \bullet G}] \, . \, ![!\interp{G}] \, . \, ![!\interp{\tau}]
\end{align*}
\ie{}, a channel over which four things can be sent: a $\interp{\sigma}$ value for the function argument, a channel which can receive an effect channel capable of simulating effects $F \bullet G$, a channel over which can be sent an effect channel capable of simulating effects $G$, and a channel which can send a $\interp{\tau}$ for the result. Thus, the encoding of a 
function receives effect handling channels which have the same form as the effect 
channels for first-order term encodings. 
 A full, formal treatment of effects in a higher-order setting, and the requirements on
the underlying calculi, is forthcoming work.

Effects systems also commonly include a (partial) ordering on effects,
which describes how effects can be
overapproximated~\cite{gifford1986effects}. For example, causal state
effects could be ordered by prefix inclusion, thus an expression $M$ with
judgement $\Gamma \vdash M : \tau, \eff{[\getF{} \, \tau]}$ might have
its effects overapproximated (via a subsumption rule) to $\Gamma
\vdash M : \tau, \eff{[\getF{} \, \tau, \putF{} \, \tau']}$. It is
possible to account for (some) subeffecting using subtyping of
sessions. Formalising this is further work.

Whilst we have embedded effects into sessions, the converse seems
possible: to embed sessions into effects. 
Nielson and Nielson previously defined an effect system for
higher-order concurrent programs which resembles some aspects
of session types~\cite{nielson1994higher}. Future work is to 
explore mutually inverse embeddings of sessions and effects. 
Relatedly, further work is to explore whether various kinds of \emph{coeffect system}
(which dualise effect systems, analysing
context and resource use~\cite{DBLP:conf/icfp/PetricekOM14})
 such as bounded linear logics, can also be embedded into session types 
or vice versa. 


\paragraph{Acknowledgements}

Thanks to Tiago Cogumbreiro and the anonymous reviewers for their
feedback. 
The work has been partially sponsored by  EPSRC EP/K011715/1, EP/K034413/1,
 and EP/L00058X/1, and EU project FP7-612985 UpScale.

\bibliography{references}

\appendix

%

\section{Subtyping and selection}
\label{subtyping-discussion} 

Our session typing system assigns selection types that include 
only the label $l$ being selected ((select) in Figure~\ref{fig:pi-calc-typing}).
Duality with branch types is provided by subtyping on selection types: 
\begin{align*}
\trule{sel} \qquad {\oplus [\tilde{l} : \tilde{S}] \, \prec \, \oplus[\tilde{l} : \tilde{S}, \tilde{l}' : \tilde{S}']}
\end{align*}
(this is a special case of the usual full subtyping rule for selection, see~\cite[\textsc{[sub-sel]}, Table 5, p. 4]{CDY2014}). Therefore, for example, the $\textit{get}$ process
could be typed: 
\begin{align*}
\trule{sub}\dfrac{
\dfrac{\Gamma, x : \tau; \Delta; \overline{c} : S \vdash P}
{\Gamma; \Delta , \overline{c} : \oplus [\getS : \, ? [\tau] . S]
\vdash \mathit{get}(c)(x).P} \quad \trule{sel}\dfrac{}{\oplus [\getS : \, ? [\tau] . S] \prec \oplus[\getS : \, ? [\tau] . S, \putS : \, ! [\tau] . S]}}
{\Gamma; \Delta , \overline{c} : \oplus [\getS : \, ? [\tau] . S, \putS : \, ![\tau] . S] \vdash \mathit{get}(c)(x).P}
\end{align*}

\noindent
However, such subtyping need only be applied when duality is being
checked, that is, when opposing endpoints of a channel are bound by channel
restriction, $\nu c . P$.  We take this approach, thus subtyping is
only used with channel restriction such that, prior to restriction,
session types can be interpreted as effect annotations
with selection types identifying effectful operations.

\ifappendix
\subsection{Session typing derivations for \textit{Store}}
\label{sec:store-types}

\subsection{Shared channel}

\begin{align*}
 \trule{acc}
 \dfrac{\Gamma; \Delta, c : S \vdash P}
       {\Gamma, k : \langle{S, \overline{S}}\rangle; \Delta \vdash \texttt{accept} \, k (c) . P}
 & 
 \trule{req}
 \dfrac{\Gamma; \Delta, c : \overline{S} \vdash P}
       {\Gamma, k : \langle{S, \overline{S}}\rangle; \Delta \vdash \texttt{request} \, k (c) . P}
\end{align*}
\else
\fi

\section{Agda encoding}
\label{sec:agda}

The Agda formalisation of our embedding defines data types of typed terms for the
effect calculus \texttt{\_,\_$\vdash$\_,\_} and session calculus \texttt{\_*\_\(\vdash\)\_},
indexed by the effects, types, and contexts terms: \\[-0.75em]
\begin{SVerbatim}
\textbf{data}_,_\(\vdash\)_,_ (eff : Effect) : (Gam : Context Type) -> Type -> (Carrier eff) -> Set \textbf{where} \(\ldots\)
\textbf{data} _*_\(\vdash\)_: (\(\Gamma\) : Context VType) -> (\(\Sigma\) : Context SType) -> (t : PType) -> Set \textbf{where} \(\ldots\)
\end{SVerbatim}
\vspace{0.25em}

\noindent
These type constructors are multi-arity infix operators. 
For the effect calculus type, the first index \texttt{eff \!\!: \!\!\!Effect} is a record
providing the effect algebra, operations, and constants, of which the 
\texttt{Carrier} field holds the type for effect annotations.
The embedding is then a function: \\[-0.75em]
\begin{SVerbatim}
embed : \textbf{forall} \{\(\Gamma\) \(\tau\) F\} -> (e : stEff , \(\Gamma\) \(\vdash\) \(\tau\) , F)
                    \hspace{0.17em}   -> (map interpT \(\Gamma\)) * ((Em , [ interpT \(\tau\) ]!\(\cdot\,\)end) , interpEff F) \(\vdash\) proc
\end{SVerbatim}
\vspace{0.25em}

\noindent
where {\small{\texttt{interpT : Type -> VType}}} maps types of the effect calculus to value types for
sessions, and {\small{\texttt{interpEff : List StateEff -> SType}}} maps
state effect annotations to session types {\small{\texttt{SType}}}. Here
the constructor \texttt{[\_]!\(\cdot\)\_} is a binary data constructor 
representing the session type for send. The intermediate embedding has the type
(which also uses the receive session type \texttt{[\_]?\(\cdot\)\_}): \\[-0.55em]
\begin{SVerbatim}
embedInterm : \textbf{forall} \{\(\Gamma\) \(\tau\) F G\} 
  -> (M : stEff , \(\Gamma\) \(\vdash\) \(\tau\) , F)
  -> (map interpT \(\Gamma\) * (((Em , [ interpT \(\tau\) ]!\(\cdot\) end) , [ sess (interpEff (F ++ G)) ]?\(\cdot\) end) 
                     \hspace{11.3em}       , [ sess (interpEff G) ]!\(\cdot\,\)end) \(\vdash\) proc
\end{SVerbatim}  
\vspace{-0.6em}

\section{Soundness proof of embedding, wrt. Figure~\ref{fig:equations} equations}
\label{sec:proofs}

\newcommand{\defEq}{\stackrel{\text{def}}{=}}

\textbf{Theorem} (Soundness). 
If
$\Gamma \vdash M \equiv N : \tau, F$ 
then
$\interp{\Gamma}; 
(r : \, !\interp{\tau}.\mathbf{end}, \textit{e} : \interp{F})
\vdash \interp{M}^{e}_r \approx \interp{N}^{e}_r$ \\[-0.6em]

\noindent
\textnormal{\textbf{Proof}} $\;$
We make use of an intermediate lemma, which we call \emph{forwarding}, where for all $M$:
\begin{align*}
& \nu \ea . (\interpA{M}^{\ei,\ea}_r \mid \du{\ea}?(c).\eo!\langle{c}\rangle.P)
\; \approx \; \interpA{M}^{\ei,\eo}_r \mid P  \quad 
\wedge \quad  \nu \ea . (\interpA{M}^{\ea,\eo}_r \mid \ei?(c).\du{\ea}!\langle{c}\rangle.P) 
\; \approx \; \interpA{M}^{\ei,\eo}_r \mid P
\end{align*}
where $c$ is not free in $P$. 
This follows by induction on the definition of the intermediate embedding.

Since $\interp{M}^{e}_r = \nu \ei, \eo . \; (\interpA{M}^{\ei,\eo}_r \mid \, \overline{\ei}!\langle{e}\rangle.{\eo}?(c))$ 
and $\interp{N}^{e}_r = \nu \ei, \eo . \; (\interpA{N}^{\ei,\eo}_r \mid \, \overline{\ei}!\langle{e}\rangle.{\eo}?(c))$ (eq.~\ref{def:top-level}) we need only consider
 $\interpA{M}^{\ei,\eo}_r \approx \interpA{N}^{\ei,\eo}_r$ \ie{}, weak bisimilarity of the intermediate
embeddings. We address each equation in turn. The relation $\defEq{}$
denotes definitional equality based on $\interpA{-}^{\ei,\eo}_r$.
\def\arraystretch{1.3}
\begin{align*}
\begin{array}{rlr}
& \hspace{-1em}\text{(unitR)} \\
& \quad\;\; \interpA{\synLet{M}{x}{x}}^{\ei,\eo}_r \\
 &
\defEq \nu \, q, \ea . \, (\interpA{M}^{\ei,\ea}_q \mid \overline{q}?(x). \ea?(c).r!\langle{x}\rangle.\overline{\eo}!\langle{c}\rangle) \\[-0.1em]
 &
\approx \nu \, \ea . \, (\interpA{M}^{\ei,\ea}_r \mid \ea?(c).\overline{\eo}!\langle{c}\rangle)
& \quad \text{\{forwarding $q \rightarrow r$\}} \\
& \approx \interpA{M}^{\ei,\eo}_r \quad \Box 
& \quad \text{\{forwarding $\ea \rightarrow \eo$\}} \\[0em]
%
%
& \hspace{-1em}\text{(unitL)} \\
& \quad\;\; \interpA{\synLet{x}{y}{M}}^{\ei,\eo}_r \\[-0.1em]
& \defEq \nu \, q, \ea . (\interpA{x}^{\ei,\ea}_q \mid \overline{q}?(y). \interpA{M}^{\ea,\eo}_r) \\[-0.1em]
& \defEq \nu \, q, \ea . (\ei?(c).q!\langle{x}\rangle.\overline{\ea}!\langle{c}\rangle
\mid \overline{q}?(y). \interpA{M}^{\ea,\eo}_r)\\[-0.1em]
& \approx \nu \, \ea . (\ei?(c).\overline{\ea}!\langle{c}\rangle
\mid \interpA{M}^{\ea,\eo}_r[x/y]) & \quad \text{\{$\beta$, structural congruence\}} \\ 
& \approx \interpA{M}^{\ei,\eo}_r[x/y] & \quad \text{\{forwarding $\ei \rightarrow \ea$}\} \\
& \approx \interpA{M [x/y]}^{\ei,\eo}_r \quad \Box & \quad \text{\{var substitution preserved by $\interpA{-}$\}}
\end{array}
\end{align*}
\vspace{-1.2em}
\begin{align*}
\begin{array}{rlr}
& \hspace{-1em}\text{(assoc)} \\
& \quad\;\; \interpA{\synLet{(\synLet{M}{x}{N})}{y}{P}}^{\ei,\eo}_r \\
& \defEq \nu \, q, \ea . \, (\interpA{\synLet{M}{x}{N}}^{\ei,\ea}_q \mid \overline{q}?(y). \interpA{P}^{\ea,\eo}_r) \\[-0.1em]
& \defEq \nu \, q, \ea . \, (\nu \, q1, \eb . \, (\interpA{M}^{\ei,\eb}_{q1} \mid \overline{q1}?(x). \interpA{N}^{\eb,\ea}_q) \mid \overline{q}?(y). \interpA{P}^{\ea,\eo}_r)
& \\
(*) \hspace{-0.75em} & \approx \nu \, q, \ea, q1, \eb . \, (\interpA{M}^{\ei,\eb}_{q1} \mid \overline{q1}?(x). \interpA{N}^{\eb,\ea}_q \mid \overline{q}?(y). \interpA{P}^{\ea,\eo}_r) 
&\hspace{-0.6em} \text{\small{\{structural congruence\}}} \\[0.35em]
%
& \quad\;\; \interpA{\synLet{M}{x}{(\synLet{N}{y}{P})}}^{\ei,\eo}_r \\
%
& \defEq \nu \, q, \ea . \, (\interpA{M}^{\ei,\ea}_q \mid \overline{q}?(x). \interpA{\synLet{N}{y}{P}}^{\ea,\eo}_r) \\
%
& \defEq \nu \, q, \ea . \, (\interpA{M}^{\ei,\ea}_q \mid \overline{q}?(x). \nu \, q1, \eb . \, (\interpA{N}^{\ea,\eb}_{q1} \mid \overline{q1}?(y). \interpA{P}^{\eb,\eo}_r)) \\ 
%
& \approx \nu \, q, \ea, q1, \eb . \, (\interpA{M}^{\ei,\ea}_q \mid \overline{q}?(x). \interpA{N}^{\ea,\eb}_{q1} \mid \overline{q1}?(y). \interpA{P}^{\eb,\eo}_r) 
  & \hspace{-1.6em} \text{\small{\{sequentiality, $x \not\in fv(P)$\}}} \\
& \approx \nu \, q, \ea, q1, \eb . \, (\interpA{M}^{\ei,\eb}_{q1} \mid \overline{q1}?(x). 
\interpA{N}^{\eb,\ea}_{q} \mid \overline{q}?(y). \interpA{P}^{\ea,\eo}_r) 
  & \hspace{-1.2em} \text{\small{\{$\alpha$, $\ea \leftrightarrow \eb, q \leftrightarrow q1$\}}} \\
& \approx (*) \quad \Box 
\end{array}
\end{align*}

\noindent
The proof of the commutativity axiom (comm) relies on an additional equation of the effect algebra, that
for all $F, G$ then $F \bullet G \equiv I \Rightarrow (F \equiv G \equiv I)$. The proof 
below is factored into an additional lemma about encodings of pure computations, given 
 after in Lemma~\ref{lem:purity}. Essentially, given the right session calculus
context $C$, then $C[\interpA{\Gamma \vdash M : \tau, I}^{\ei,\eo}_r] \approx C[\ei?(c).\du{\eo}!\langle{c}\rangle
\mid \interp{M}_r]$. That is, a pure computation can be factored into a pure encoding and a forwarding
from the input effect-channel-carrying channel to the output channel. 

\begin{align*}
\begin{array}{rlr}
& \hspace{-1em} \text{(comm) where $\Gamma \vdash M : \tau_1, I$} \\
& \quad\; \interpA{\synLet{M}{x}{(\synLet{N}{y}{P})}}^{\ei,\eo}_r \\
& \defEq \nu \, q, \ea . \, (\interpA{M}^{\ei,\ea}_q \mid \overline{q}?(x). \interpA{\synLet{N}{y}{P}}^{\ea,\eo}_r) \\
& \defEq \nu \, q, \ea . \, (\interpA{M}^{\ei,\ea}_q \mid \overline{q}?(x). \nu \, q1, \eb . \, (\interpA{N}^{\ea,\eb}_{q1} \mid \overline{q1}?(y). \interpA{P}^{\eb,\eo}_r)) \\ 
& \approx \nu \, q, \ea, q1, \eb . \, (\interpA{M}^{\ei,\ea}_q \mid \interpA{N}^{\ea,\eb}_{q1} \mid \overline{q}?(x).\overline{q1}?(y). \interpA{P}^{\eb,\eo}_r) 
  & \hspace{-2.4em} \text{\small{\{sequentiality, $x \not\in fv(N)$\}}} \\
& \approx \nu \, q, \ea, q1, \eb . \, (\ei?(c).\overline{\ea}!\langle{c}\rangle \mid \interp{M}_q \mid \interpA{N}^{\ea,\eb}_{q1} \mid \overline{q}?(x).\overline{q1}?(y). \interpA{P}^{\eb,\eo}_r) 
& \hspace{-1.4em} \text{\small{\{purity lemma~\ref{lem:purity} on $M$\}}} \\
\hspace{-0.5em}(*) \hspace{-0.64em} & \approx \nu \, q, q1, \eb . \, (\interp{M}_q \mid \interpA{N}^{\ei,\eb}_{q1} \mid \overline{q}?(x).\overline{q1}?(y). \interpA{P}^{\eb,\eo}_r) 
& \hspace{-1em} \text{\small{\{forwarding $\ei \rightarrow \ea$\}}} \\[0.8em]
& \quad\; \interpA{\synLet{N}{y}{(\synLet{M}{x}{P})}}^{\ei,\eo}_r \\
& \defEq \nu \, q, \ea . \, (\interpA{N}^{\ei,\ea}_q \mid \overline{q}?(y). \interpA{\synLet{M}{x}{P}}^{\ea,\eo}_r) \\
& \defEq \nu \, q, \ea . \, (\interpA{N}^{\ei,\ea}_q \mid \overline{q}?(y).
\nu \, q1, \eb . \, (\interpA{M}^{\ea,\eb}_{q1} \mid \overline{q1}?(x). \interpA{P}^{\eb,\eo}_r)) \\
& \approx \nu \, q, \ea, q1, \eb . \, (\interpA{N}^{\ei,\ea}_q \mid 
\interpA{M}^{\ea,\eb}_{q1} \mid \overline{q}?(y).\overline{q1}?(x). \interpA{P}^{\eb,\eo}_r)
& \hspace{-2.4em} \text{\small{\{sequentiality, $y \not\in fv(M)$\}}} \\
& \approx \nu \, q, \ea, q1, \eb . \, (\interpA{N}^{\ei,\ea}_q \mid 
\ea?(c).\overline{\eb}!\langle{c}\rangle \mid \interp{M}_{q1} \mid \overline{q}?(y).\overline{q1}?(x). \interpA{P}^{\eb,\eo}_r) \;\;
& \hspace{-1.4em} \text{\small{\{purity lemma~\ref{lem:purity} on $M$\}}} \\
& \approx \nu \, q, q1, \ea . \, (\interpA{N}^{\ei,\ea}_q \mid 
\interp{M}_{q1} \mid \overline{q}?(y).\overline{q1}?(x). \interpA{P}^{\ea,\eo}_r)
& \hspace{-1em} \text{\small{\{forwarding $\ea \rightarrow \eb$\}}} \\
& \equiv \nu \, q, q1, \ea . \, (\interp{M}_{q1} \mid \interpA{N}^{\ei,\ea}_q \mid 
\overline{q}?(y).\overline{q1}?(x). \interpA{P}^{\ea,\eo}_r)
& \hspace{-1.4em} \text{\small{\{structural congruence\}}} \\
& \stackrel{\alpha}{\equiv} \nu \, q, q1, \eb . \, (\interp{M}_{q} \mid \interpA{N}^{\ei,\eb}_{q1} \mid 
\overline{q1}?(y).\overline{q}?(x). \interpA{P}^{\eb,\eo}_r)
& \hspace{-1em} \text{\small{\{$\alpha$, $q \leftrightarrow q1$, $\ea \leftrightarrow \eb$\}}} \\
& \approx \nu \, q, q1, \eb . \, (\interp{M}_{q} \mid \interpA{N}^{\ei,\eb}_{q1} \mid 
\overline{q}?(x).\overline{q1}?(y).\interpA{P}^{\eb,\eo}_r)
& \hspace{-1em} \text{\small{\{reorder recv.\}}} \\
& \approx (*) \quad \Box
\end{array}
\end{align*}

\begin{lemma}[Pure encodings, in context]
If an effect system has the property that
$\forall F, G .   (F \bullet G \equiv I) \, \Rightarrow \, (F \equiv G \equiv I)$ then, 
for all $M$, $\Gamma$, $\tau, P, Q$ it follows that:
$$
\interpA{\Gamma \vdash M : \tau, I}^{\ei,\eo}_r \mid \eo?(c).P \mid \du{r}?(x).Q 
\;\;\; \approx \;\;\;
\ei?(c).\du{\eo}!\langle{c}\rangle \mid \interp{M}_r \mid \eo?(c).P \mid \du{r}?(x).Q
$$
\label{lem:purity}
\vspace{-1.6em}
\end{lemma}

\begin{proof}[\textnormal{\textbf{Proof}}]
By induction over type-and-effect derivations with a pure effect in the conclusion. 
\begin{itemize}
\item (var) 
\vspace{-1.6em}
\begin{align*}
\begin{array}{ll}
\begin{array}{rl}
  & \interpA{\Gamma \vdash x : I}^{\ei,\eo}_r \mid \eo?(c).P \mid \du{r}?(x).Q \\
= & \ei?(c).r!\langle{x}\rangle.\du{\eo}!\langle{c}\rangle \mid \eo?(c).P \mid \du{r}?(x).Q \\
\xrightarrow{\ei(c)} & r!\langle{x}\rangle.\du{\eo}!\langle{c}\rangle \mid \eo?(c).P \mid \du{r}?(x).Q  \\
\xrightarrow{\tau} & \du{\eo}!\langle{c}\rangle \mid \eo?(c).P \mid Q \\
\xrightarrow{\tau} & P \mid Q 
\end{array}
&
\begin{array}{rl}
  & \ei?(c).\du{\eo}!\langle{c}\rangle \mid \interp{x}_r \mid \eo?(c).P \mid \du{r}?(x).Q \\
= & \ei?(c).\du{\eo}!\langle{c}\rangle \mid r!\langle{x}\rangle \mid \eo?(c).P \mid \du{r}?(x).Q  \\
\xrightarrow{\ei(c)} & \du{\eo}!\langle{c}\rangle \mid r!\langle{x}\rangle \mid \eo?(c).P \mid \du{r}?(x).Q  \\
\xrightarrow{\tau} & r!\langle{x}\rangle \mid P \mid \du{r}?(x).Q \\
\xrightarrow{\tau} & P \mid Q 
\end{array}
\end{array}
\end{align*}

\item (const) Similar to (var), where $\interpA{\Gamma \vdash C : C_\tau, I}^{\ei,\eo}_r = \ei?(c).\interp{C}_r.\du{\eo}!\langle{c}\rangle$ which has the same shape as the (var) encoding and thus follows a similar proof to the above. 
\item (op) Similar to the above, where $\interpA{\Gamma \vdash op \, M : Op_\tau, I}^{\ei,\eo}_r = \ei?(c).\interp{op \, M}_r.\du{\eo}!\langle{c}\rangle$ which has the same shape as (var) and (const) encodings.
\item (let) By the additional requirement that $\forall F, G . (F \bullet G \equiv I) \Rightarrow (F \equiv G \equiv I)$ then the encoding of a type-and-effect derivation rooted in the (let) rule (with pure effect) 
is necessarily of the form: 
\begin{equation*}
\interpA{\Gamma \vdash \synLet{M}{x}{N} : \tau, I}^{\ei,\eo}_r = \nu q, \ea . (\interpA{\Gamma \vdash M : \sigma, I}^{\ei,\ea}_q \mid \overline{q}?(x). \interpA{\Gamma, x : \sigma \vdash N : \tau, I}^{\ea,\eo}_r)
\end{equation*}
From the premises, the inductive hypothesis are:
\begin{align*}
(A) \;\; & \interpA{\Gamma \vdash M : \sigma, I}^{\ei,\ea}_q \mid \ea?(c).P \mid \du{q}?(x).Q 
\;\; \approx \;\;
\ei?(c).\du{\ea}!\langle{c}\rangle \mid \interp{M}_q \mid \ea?(c).P \mid \du{q}?(x).Q \\
(B) \;\; & \interpA{\Gamma, x : \sigma \vdash N : \tau, I}^{\ea,\eo}_r \mid \eo?(c).P' \mid \du{r}?(x).Q'
\;\; \approx \;\;
\ea?(c).\du{\eo}!\langle{c}\rangle \mid \interp{N}_r \mid \eo?(c).P' \mid \du{r}?(x).Q' 
\end{align*}
From these the following weak bisimilarity holds: 
\begin{align*}
\begin{array}{rl}
& \nu q, \ea . (\interpA{\Gamma \vdash M : \sigma, I}^{\ei,\ea}_q \mid \overline{q}?(x). \interpA{\Gamma, x : \sigma \vdash N : \tau, I}^{\ea,\eo}_r) \mid \eo?(c).P \mid \du{r}?(x).Q \\
\stackrel{(B)}{\approx} & \nu q, \ea . (\interpA{\Gamma \vdash M : \sigma, I}^{\ei,\ea}_q \mid \overline{q}?(x). 
(\ea?(c).\du{\eo}!\langle{c}\rangle \mid \interp{N}_r \mid \eo?(c).P \mid \du{r}?(x).Q)\\
\stackrel{fwd}{\approx} & \nu q . (\interpA{\Gamma \vdash M : \sigma, I}^{\ei,\eo}_q \mid \overline{q}?(x). 
(\interp{N}_r \mid \eo?(c).P \mid \du{r}?(x).Q) \\
\stackrel{fwd}{\approx} & \nu q, \ea . (\interpA{\Gamma \vdash M : \sigma, I}^{\ei,\ea}_q \mid \ea?(c).\du{\eo}!\langle{c}\rangle \mid \overline{q}?(x). 
(\interp{N}_r \mid \eo?(c).P \mid \du{r}?(x).Q)) \\
\stackrel{(A)}{\approx} & \nu q, \ea . (
\ei?(c).\du{\ea}!\langle{c}\rangle \mid \interp{M}_q \mid \ea?(c).\du{\eo}!\langle{c}\rangle \mid \du{q}?(x).(\interp{N}_r \mid \eo?(c).P \mid \du{r}?(x).Q))\\ 
\xrightarrow{\ei(c)} & \nu q, \ea . (
\du{\ea}!\langle{c}\rangle \mid \interp{M}_q \mid \ea?(c).\du{\eo}!\langle{c}\rangle \mid \du{q}?(x).(\interp{N}_r \mid \eo?(c).P \mid \du{r}?(x).Q)) \\
\xrightarrow{\tau} & \nu q . (
\interp{M}_q \mid \du{\eo}!\langle{c}\rangle \mid \du{q}?(x).(\interp{N}_r \mid \eo?(c).P \mid \du{r}?(x).Q)) \\
\approx & \nu q . (\interp{M}_q \mid \du{q}?(x).(\interp{N}_r \mid P \mid \du{r}?(x).Q))\\
\equiv & \nu q . (\interp{M}_q \mid \du{q}?(x).\interp{N}_r) \mid P \mid \du{r}?(x).Q \\
\equiv & \interp{\synLet{M}{x}{N}}_r \mid P \mid \du{r}?(x).Q  \\[0.6em]
& \ei?(c).\du{\eo}!\langle{c}\rangle \mid \interp{\synLet{M}{x}{N}}_r \mid \eo?(c).P \mid \du{r}?(x).Q \\[-0.2em]
\xrightarrow{\ei(c)} & \du{\eo}!\langle{c}\rangle \mid \interp{\synLet{M}{x}{N}}_r \mid \eo?(c).P \mid \du{r}?(x).Q \\
\xrightarrow{\tau} & \interp{\synLet{M}{x}{N}}_r \mid P \mid \du{r}?(x).Q 
\end{array}
\end{align*}
\end{itemize}
\vspace{-1.6em}
\end{proof}
\end{document}